\documentclass[a4paper]{article}

\title{Discrete dividends after maturity  adjust the stock and strike prices}
\author{Kevin W.\ Lu\thanks{Research School of Finance, Actuarial Studies \& Statistics, Australian National University, Canberra,  ACT 0200, Australia.  Email: \href{mailto:kevin.lu@anu.edu.au}{kevin.lu@anu.edu.au}.}}

\usepackage{natbib}
\usepackage{graphicx}
\usepackage{amsmath,amsfonts,amssymb,amsthm}
\usepackage[colorlinks,citecolor=red,linkcolor=blue]{hyperref}
\usepackage{float}
\usepackage{mathrsfs}
\usepackage{fancyhdr} 
\usepackage{microtype} 
\usepackage{datetime2} 
\usepackage{siunitx}
\newcommand{\mytoday}{%
    \number\day\space
    \ifcase\month\or
    January\or February\or March\or April\or May\or June\or
    July\or August\or September\or October\or November\or December%
    \fi\space
    \number\year
}
\date{\mytoday}

\numberwithin{equation}{section} 
\allowdisplaybreaks[4] 
\let\originalleft\left
\let\originalright\right
\def\left#1{\mathopen{}\originalleft#1}
\def\right#1{\originalright#1\mathclose{}}

\fancypagestyle{plain}{
    \fancyhf{}
    
    \fancyfoot[C]{\thepage}
    }

\newtheorem{theorem}{Theorem}[section]
\newtheorem{corollary}[theorem]{Corollary}
\newtheorem{lemma}[theorem]{Lemma}
\newtheorem{proposition}[theorem]{Proposition}

\theoremstyle{definition}

\newtheorem{example}[theorem]{Example}
\theoremstyle{remark}
\newtheorem{remark}[theorem]{Remark}

\newcommand{\halmos}{\quad\hfill\mbox{$\Box$}}

\newcommand{\EE}{\mathbb{E}}
\newcommand{\FF}{\mathbb{F}}

\newcommand{\RR}{\mathbb{R}}

\newcommand{\PP}{\mathbb{P}}

\newcommand{\FFF}{{\cal F}}

\newcommand{\TTT}{{\cal T}}

\newcommand{\wt}{\widetilde}

\newcommand{\given}{{\,\vert\,}}

\newcommand{\wh}{\widehat}

\newcommand{\name}{\operatorname}

\begin{document}

    \maketitle
    
    \begin{abstract}
        
        The standard method to price European calls on a discrete dividend-paying stock is to  subtract the present value of the dividends from the initial stock price in the Black-Scholes formula. However, when there are dividends after maturity, this is inconsistent with the model. Within the escrowed dividend model, we highlight an extension of the Black-Scholes formula in which those  dividends after maturity adjust both the stock price and strike price, allowing for calls over all maturities to be priced in a model-consistent way. As a related application to American calls with a single dividend before maturity, we establish a neglected case where  it is always optimal to early exercise and derive an extension of the Roll-Geske-Whaley formula when there are dividends after maturity, including fully characterising the optimal exercise policy.

    \end{abstract}

    \section{Introduction}

    The price of a European call on a stock that pays discrete dividends of fixed amounts is said  to be obtained by subtracting the present value of the dividends from the initial stock price in  the Black-Scholes formula, so that
    \begin{align}
        c_0 = \overline{S}_0N(d_1) - Ke^{-rT}N(d_2),\label{bs}
    \end{align}
    where
    \begin{align*}
        d_1 = \frac{\log({\overline{S}_0}/K) +(r+\frac 12 \sigma^2)T}{\sigma\sqrt{T}},\quad d_2 = d_1 -\sigma\sqrt{T}, \quad  \overline{S}_0= S_0 -\sum_{i=1}^n D_ie^{-rt_i}
    \end{align*}
    with the usual notation in which only dividends before maturity are included. Popularised by \cite{hull} and its earlier editions, this standard approach is commonly taught and appears in many other textbooks including \cite{jarc24,Kwo08,das}, although not in all \cite{Bjo09,shreve}.  
    To justify this, the stock price is assumed to be a sum of a risky and riskless component. This is known as the escrowed dividend model and defined more precisely in Section \ref{sec2}. When further assuming that the risky component is  geometric Brownian motion (GBM), we call this the Black-Scholes escrowed dividend model. Then the standard Black-Scholes formula can be applied to price a call on the risky component, from which the price of a call on the original stock is obtained giving \eqref{bs}.

    We naturally take the escrowed dividend model as a single model defined with a fixed time horizon for  which derivatives of all earlier maturities  are priced consistently, rather than understanding it as simply the application of \eqref{bs} or having different models for the stock price for each maturity time as in  \cite{Fri02,HauHauLew03,hull}.
    It turns out that in the Black-Scholes escrowed dividend model, call prices also depend on the dividends paid after maturity, which adjust the stock  and strike prices. In these situations, using \eqref{bs} would be incorrect. The error may be reasonably large, such as for a deeply out-of-the-money call with a discrete dividend shortly after maturity, but it may be modest in many other cases.

    The purpose of this paper is to account for dividends after maturity in deriving an extension of the Black-Scholes formula  for the price of European calls, and the Roll-Geske-Whaley (RGW)  formula for the price of American calls with a single discrete dividend before maturity.  The method of proof is based on elementary change of variable arguments, although more involved in the latter situation. The extended Black-Scholes formula is given  in Theorem \ref{mainthm} and reduces to \eqref{bs} when there are no dividends  after maturity. We  show that standard no arbitrage results hold in Proposition \ref{narel}, including the European option bounds and put-call parity, and do {not} depend on dividends after maturity. But we also obtain improved call and put upper bounds, which do depend on dividends after maturity. Then we reveal a  case  neglected in the literature about the RGW formula, namely $D_1 \geq K$, where it is always optimal to exercise the call just prior to the dividend time. Using this, the extended RGW formula is given in  Theorem \ref{theorem2}, and we also fully characterise the optimal exercise policy. Lastly, we discuss how the extension of the Black-Scholes formula resolves erroneous claims in \cite{Hau07,Fri02}  that the RGW formula admits arbitrage.

    We place the escrowed dividend model and our results in the context of other approaches in the literature for dealing with discrete dividends. To begin with,  adjustments to stock and strike are considered in \cite{BosVan02} and to volatility in \cite{BosGaiShe03,BenVor02}, but unlike these methods and others \cite{BurGuo12,Fri02}, which are justified based on heuristic arguments and numerical experiments, we take as the starting point the escrowed dividend model, a logically consistent no arbitrage model from which all pricing formulas and other results are derived. Thus, we take a model-consistent approach and  avoid the ad hoc approach of some previous works.

    The escrowed dividend model is  equivalent to the models for discrete dividend-paying stocks in \cite{Bue10,Kla15} and \cite{BerMai15}, despite their apparently different  formulations, as  explained in Remark~\ref{equivlem}.  In \cite{BerMai15}, which is otherwise about a tree-based pricing method, it is noted in Remark 4 that analytical pricing formulas for European derivatives can be extended to the case where there are dividends before maturity, assuming there are none after maturity. In contrast, we obtain analytical pricing formulas in an equivalent model specifically when there are dividends after maturity.

    The main alternative model to the escrowed dividend model is the piecewise GBM model (also called the piecewise lognormal model \cite{VelNi06}), where the stock is a GBM that drops by the dividend amount. This model is surprisingly complicated because the dividend amount must be a function of the stock price, rather than a constant, in order for the stock price to be positive. \cite{HauHauLew03} studied the dividend policies required for this positivity and developed integral methods for option pricing, while other numerical methods for this model include analytical approximations with \cite{MatDilFer09} and without  \cite{DaiLyu09} arbitrary precision,  tree-based methods \cite{VelNi06}, and series methods  \cite{VeiWys09}.  Appearing as Model 3 in \cite{Fri02}, the piecewise GBM model is advocated for on the basis that it better reflects stock prices in reality,  but this is based on intuition rather than empirical evidence. Nevertheless, it is widely regarded in the literature as the benchmark or true model, and numerical results from other models and methods, including the escrowed dividend model, are  often compared to it \cite{DaiLyu09,HauHauLew03,VelNi06,VeiWys09}. Indeed, \cite{BurGuo12} summarises these comparisons, and \cite{BosVan02} considers an adjustment for the escrowed dividend model to approximately reproduce the call prices under the piecewise GBM model. However, there is no particular reason why the escrowed dividend model could  not be taken as the true model instead.   While dividends after maturity do not affect call prices under the piecewise GBM model, the escrowed dividend model has the advantage of analytical  formulas. Contrary to the preference of \cite{HauHauLew03}, we discuss further reasons to prefer the escrowed dividend model in Remarks \ref{rem1}, \ref{rem2} and \ref{pcfail}, including that it is the most natural way to treat discrete dividends.

    Lastly, another common modelling assumption from \cite{Bjo09,shreve} is to assume that dividends are a fixed proportion of the stock price at the dividend time, rather than  a fixed amount. In this case, the call price does not depend on  dividends after maturity. The model in \cite{Kla15,Bue10} allows for affine dividends, which generalises proportional and fixed  dividends.

    We now comment on our contributions. The extension of the Black-Scholes formula is not new. It appears in \cite{Kla15}, who writes that ``pricing of an option of any maturity will depend on all cash dividends, even those after maturity" and traces back to \cite{Bue10,BerBueFerJorLamOve07}. However, the result and its consequences appear  not to be  widely appreciated. Most of the standard no arbitrage results are of course well-known from textbook arbitrage arguments, though they are presented here to highlight the consistency of the model and because Proposition \ref{narel} (vi) is needed for the extended RGW formula.
    To our knowledge, the improved upper bounds and the  applications to the RGW formula, including the always early exercise case and the extended formula for dividends after maturity, are new contributions.   Additional comments about originality are given in Remarks \ref{ogrem1},  \ref{caserem} and \ref{ogrem2}.

    The paper is structured as follows. In Section \ref{sec2}, we introduce the escrowed dividend model,  and prove the extended Black-Scholes formula and  the standard no arbitrage results. In Section \ref{sec3}, we firstly begin with a lemma that fully characterises the optimal early exercise policy, secondly apply it to uncover an always early exercise case in the RGW formula, thirdly prove the extended RGW formula, and lastly comment on how claims that the RGW formula admits arbitrage are incorrect. In Section \ref{sec4}, we conclude with the benefits of the escrowed dividend model in providing a package of consistent analytical pricing formulas.

    \section{The escrowed dividend model}\label{sec2}
    
    For a time horizon $\overline{T}$ and  filtration  $\FF=(\FFF_t)_{t\in[0,\overline{T}]}$, we suppose the  continuous  time market model has 2 traded assets, the risk-free asset with constant risk-free rate $r$ and  a discrete dividend-paying stock with price process $S=(S_t)_{t\in[0,\overline{T}] }$ adapted to $\FF$. The stock pays dividends of fixed amounts $D_1,\dots,D_n$ at times $t_1,\dots, t_n\in(0,\overline{T}]$, and at each dividend time $t_i$, it is assumed that a portfolio receives the value $D_i$ per unit of stock held. For $s\leq t$ and $u\geq 0$, we introduce the notation
    \begin{align*}
        \overline{D}_{s,t\given u}  := \sum_{\substack{i=1\\t_i\in(s,t] }}^n D_i e^{-r(t_{i}-u)},
    \end{align*}
    which is the present or future value at time $u$ of the dividends paid at times $(s,t]$, and let $\overline{D}_t:= \overline{D}_{t,\overline{T}\given t} $. 
    
    Let $\PP$ be the real world measure and $\wh\PP$ be the risk-neutral measure.  We assume  that the stock price process under $\PP$ satisfies 
    \begin{align}
        S_t = \overline{S}_t + \overline{D}_t,\quad t\in[0,\overline{T}],\label{pp}
    \end{align}
    where $\overline S =(\overline S_t)_{t\in[0,\overline{T}]}$ is a stochastic process representing the risky component of the stock, and  $\overline D =(\overline D_t)_{t\in[0,\overline{T}]}$ is the deterministic process representing the riskless component associated with the dividends, and this is called the \emph{escrowed dividend model}. If it is further assumed  that $\overline{S}\sim \name{GBM}(\mu,\sigma)$ is a geometric Brownian motion with initial value $\overline{S}_0>0$, drift $\mu\in\RR$ and volatility $\sigma>0$ under $\PP$, and $\FF$ is the corresponding Brownian filtration, then it is called the \emph{Black-Scholes escrowed dividend model}.

    No arbitrage is ensured by the existence of an EMM, which under discrete dividends is defined as an  equivalent measure $\wh\PP$  such that the discounted value of a self-financing portfolio with 1 unit of stock that reinvests the dividends into the risk-free asset is a $\wh \PP$-martingale (see \cite[Rule 9.19]{EK19}). By the above definitions,  an equivalent measure $\wh\PP$ being an EMM is equivalent to  the discounted risky component  $\wt {\overline S}=(e^{-rt}{\overline S}_t)_{t\in[0,\overline{T}]}$ being a $\wh \PP$-martingale. Consequently by Girsanov's theorem, the stock price process under $\wh\PP$ satisfies \eqref{pp}, where ${\overline S}\sim\name{GBM}(r,\sigma)$.
    
    Let $c_t$, $p_t$ and $C_t$ be the prices at time $t$ of a European call, European put and American call with strike price $K>0$ and maturity time $T\in(0,\overline{T}]$, respectively.

    \begin{remark}\label{rem1}
         
        The escrowed dividend model appears to originate from  \cite{Rol77}. 
        As previously noted, a common interpretation of the   model is that  \eqref{pp}  holds up to   the maturity time $T$, rather than a fixed $\overline{T}$. This has been discussed and criticised in  \cite{HauHauLew03,Kla15,VelNi06} for leading to  logically inconsistent models when pricing derivatives with different maturities. Indeed, this would be a method not a model and is another reason why the piecewise GBM model is often taken as the true model. 
         
        But clearly by \eqref{pp}, the escrowed dividend model is defined up to $\overline{T}$, and accordingly the 4 criticisms in  \cite{HauHauLew03}, that it is logically inconsistent, can exhibit negative prices, is specific to GBM, and leads to arbitrage, do not apply. In particular, if $D_1,\dots,D_n\geq0$ but otherwise unrestricted, then by \eqref{pp}, $S$ is positive as a sum of positive and nonnegative processes, although it must be the case that
        \begin{align}
            S_0 =  \overline{S}_0 + \overline{D}_{0}, \quad  \overline{S}_0> 0.\label{initial}
        \end{align}
        Moreover, $\overline{S}$ can be specified by the $\wh\PP$-dynamics of any stock price process, not just GBM, for example, that of the Heston model. The claims of arbitrage are addressed in Section \ref{secarb}.   \halmos
    \end{remark}

    \begin{remark}\label{rem2}
        Note that the form  \eqref{pp} holds without loss of generality, while  the restriction that $\wt {\overline S}$ be a  $\wh \PP$-martingale is equivalent to the existence of an EMM to ensure no arbitrage. Thus,  consistent with a point \cite[pg 3--4]{Bue10} makes for an equivalent model (see Remark \ref{equivlem}), the escrowed dividend model should be understood as a consequence, rather than an assumption. Under no arbitrage, it follows that assuming that $\overline{S}$ is a GBM is the most  natural and simple extension of the Black-Scholes  model to the discrete dividend setting. Since the above argument showing  that $\wt {\overline S}$ is a $\wh \PP$-martingale is the same when $D_1,\dots, D_n$ are random, it follows that the  piecewise GBM model must also be representable in this form, but for some other more complicated  and less natural process $\overline{S}$, providing another reason for preferring the escrowed dividend model.    \halmos
    \end{remark}

    \subsection{Black-Scholes formula with dividends after maturity}\label{sec21}

    Many formulas will involve substitutions into other formulas, so to clarify notation, we will generally use lower case to denote generic function arguments, and upper case to denote the specific values that are substituted. Distinct from the  price of a call,  let  the standard Black-Scholes formula for a non-dividend-paying stock as a function of $s>0$, $k>0$, $t\in[0,T)$, $T\in(0,\overline {T}]$ be denoted
    \begin{align*}
        c^{\name{BS}}(s,k,t,T) := sN(d_1(s,k,t,T)) - ke^{-r(T-t)}N(d_2(s,k,t,T)) ,
    \end{align*}
    where
    \begin{align*}
        d_1(s,k,t,T) &:= \frac{\log(s/k) +(r+\frac 12 \sigma^2)(T-t)}{\sigma\sqrt{T-t}},\\
        d_2(s,k,t,T) &:= d_1(s,k,t,T) -\sigma\sqrt{T-t},
    \end{align*}
    $N(x)$ is the standard normal cdf at $x$, $r\in\RR$ is the risk-free rate and $\sigma>0$ is the volatility.  Also, for $k\in\RR$, define
    \begin{align}
        c^*(s,k,t,T):=    \begin{cases}
            c^{\name{BS}}(s,k,t,T) & \text{if $k> 0$,}\\
            s-ke^{-r(T-t)} & \text{if $k\leq 0$.}
        \end{cases} \label{exp}
    \end{align}

    We now give the extension of the Black-Scholes formula for a European call  when there are dividends after maturity in terms of the function  \eqref{exp}.

    \begin{theorem} \label{mainthm}
        In the Black-Scholes escrowed dividend model, the price at time $t\in[0,T)$ of a European call with strike price $K>0$ and maturity time $T\in(0,\overline{T}]$ is
        \begin{align}
            c_{t}=c^*(\overline{S}_t,K_T,t,T),\label{main}
        \end{align}
        where  $\overline{S}_t = S_t- \overline{D}_t$, $K_T := K - \overline{D}_{T}$.
       \end{theorem}

    \begin{proof}
        Writing the  Black-Scholes formula and  the price of a forward contract as a generic identity, we have for any $r\in\RR$, $\sigma>0$, $k\in\RR$, $0\leq t<T$, if $S'\sim \name{GBM}(r,\sigma)$ with initial value $S'_0>0$, then
        \begin{align}
            e^{-r(T-t)} \EE[(S'_T-k)^+\given\FFF_t ] =c^*(S'_t,k,t,T). \label{gen}
        \end{align}

        By the risk-neutral pricing formula, the European call price is
        \begin{align}
            c_t &=e^{-r(T-t)}\wh\EE[(S_T -K)^+\given\FFF_t] \label{line1} \\
            & = e^{-r(T-t)}\wh\EE[(\overline{S}_T -K_T )^+\given\FFF_t]  \label{line2} \\
            &=c^*(\overline{S}_t,K_T,t,T),\label{line3}
        \end{align}
        where  the expectation is evaluated by applying \eqref{gen} under $\wh\PP$  with the  substitutions $S'=\overline{S} \sim\name{GBM}(r,\sigma)$ and $k = K_T$.
    \end{proof}

    \begin{corollary}\label{cor}
        In the setting of Theorem \ref{mainthm}, the European call price at time 0 is
        \begin{align*}
            \hspace{-0.5em}  c_0 = \begin{cases}
                (S_0 -\overline{D}_{0} ) N(d_1) -( K-\overline{D}_{T})e^{-rT}N(d_2)  & \text{  \hspace{-1em}   if $K > \overline{D}_{T}$,}\\
                S_0 -\overline{D}_{0,T\given 0} -Ke^{-rT} & \text{  \hspace{-1em}   if $ K \leq \overline{D}_{T}$,}
            \end{cases}
        \end{align*}
        where
        \begin{align}
            d_1 = \frac{\log\left( \displaystyle{\frac{S_0 -\overline{D}_{0}}{ K-\overline{D}_{T}}}\right) +(r+\frac 12 \sigma^2)T}{\sigma\sqrt{T}},\quad d_2 = d_1 -\sigma\sqrt{T},\label{d1d2}
        \end{align}
        and if $T\in(t_n,\overline{T}]$, then the price reduces to \eqref{bs}.
    \end{corollary}
    
    Recalling that $\overline{D}_{0}=\overline{D}_{0,\overline{T}\given 0}$, $\overline{D}_{T}=\overline{D}_{T,\overline{T}\given T}$, dividends before maturity adjust  the stock price, while dividends after maturity adjust both the stock and  strike prices and thereby affect the call price.

    \begin{remark}\label{ogrem1}
     Theorem \ref{mainthm} appears in \cite[Section 3.1]{Kla15} and \cite[Result 3.1]{Bue10}. The idea is also mentioned in \cite[pg 266]{VelNi06} without  formulas. Despite  claims  in \cite[pg 236]{MatDilFer09}, \cite[pg 37]{HauHauLew03} and  \cite[pg 39]{BurGuo12}, \eqref{bs} is the true call price when there are no dividends after maturity, not an approximation,  although it may be considered as an approximation to the piecewise GBM model in a loose sense.     In \cite[pg 157]{BosVan02}, the escrowed dividend model is criticised for failing to adjust the strike  price when there is a dividend just prior to maturity, since that is  effectively what the dividend should do. However, using \eqref{main},  the strike adjustment is precisely what happens. \halmos
    \end{remark}

    \begin{example}\label{eg1}
        Suppose the Black-Scholes escrowed dividend model holds for a stock with price $S_0 = 100$, volatility $\sigma = 0.3$, which pays dividends $D_1=3$ at $t_1=0.5$ and $D_2$ at $t_2=1.5$,  the risk-free rate is $r=0.06$, and the current time is $t=0$. Consider a European call with strike price $K = 130$ and maturity time  $T = 1$. By \eqref{initial}, $S$ is positive is equivalent to $D_2\in[0,(S_0-D_1e^{-rt_1})e^{rt_2} )=[0, 106.2319)$.

        Figure \ref{fig} shows the price of a European call $c_0$ as a function of the dividend amount after maturity $D_2$ computed by the correct formula \eqref{main} and incorrect formula  \eqref{bs}, where only the dividends before maturity are included. The corresponding  implied volatility is also shown, computed using the correct formula. We see that the error increases as $D_2$ gets larger.   \halmos    
        
        \begin{figure}[htpb]
            \begin{center}
                \includegraphics{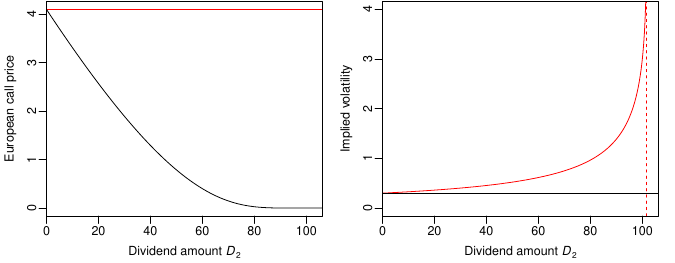}
                \caption{The European call price $c_0$ (left) and implied volatility (right) as a function of the dividend amount $D_2$ computed by the correct formula \eqref{main} (black) and  incorrect formula \eqref{bs} (red).}\label{fig}
            \end{center}
        \end{figure}
    \end{example}
    
    \begin{example}
        Consider the situation in Example \ref{eg1}  but with $T=0.1$, $t_1=0.2$, $D_2=0$ instead. This is an example of a deeply out-of-the-money call with a  dividend shortly after maturity. The correct call price is 0.0087 by  \eqref{main}, while the incorrect price is 0.0113 by  \eqref{bs}, which is  a 31\% overestimate. In contrast to Example \ref{eg1}, this is a reasonably large error  without relying on unrealistically large dividends. However, if the strike price were $K=100$, so that the call is at-the-money, then the prices are  3.9667 and 4.0786, respectively, which is only a 3\% overestimate.  \halmos
    \end{example}

    \subsection{Standard no arbitrage results with dividends after maturity}\label{sec22}
    
    While it is clear from \eqref{initial} that  dividends including after maturity may generally affect the call price, it is not universally true that they do. Specifically, many of the standard no arbitrage results, such as option bounds, remain valid and do not depend on dividends after maturity. Of course, no arbitrage results need to be model-free. This is in the sense that dividends are paid according to the escrowed dividend model, but no substantial modelling assumption is  placed on $\overline{S}$, so the following results assume the escrowed dividend model but not the Black-Scholes escrowed dividend model. Some results rely on $\overline{S}$ being nonnegative, but this is not a restrictive assumption. As noted in Remark \ref{rem1}, this ensures the stock price is positive, and any non-dividend model with nonnegative stock price extended to an escrowed dividend model would satisfy this.
    
    We use the convention that if there is indifference between exercising and not exercising, it is considered the latter.

    \begin{proposition} \label{narel} 
        Let the strike price be $K>0$  and  the maturity time be $T\in(0,\overline{T}]$. In the escrowed dividend model, at time $t\in[0,T]$, we have the following  results:
        \begin{enumerate}
            \item[(i)] For the European call, if  $\overline{S}$ is nonnegative and $K_T\leq 0$, then $c_t$ is given by \eqref{main}. Furthermore, if  $\overline{S}$ is positive and $K_T\leq 0$, then exercise at maturity has probability 1.

            \item[(ii)]  The forward price at time $t$ is
            \begin{align}
                F(t,T) = (S_t-\overline{D}_{t,T\given t})e^{r(T-t)}.\label{fp}
            \end{align}
            \item[(iii)] The put-call parity is
            \begin{align}
                c_t-p_t = S_t -\overline{D}_{t,{T}\given t} -Ke^{-r(T-t)}.
            \end{align}

            \item[(iv)] Assuming $S$ is nonnegative for the upper bound, the European call satisfies
            \begin{align}
                ( S_t -\overline{D}_{t,T\given t}-Ke^{-r(T-t)})^+\leq c_t\leq   S_t -\overline{D}_{t,T\given t}.\label{usualbound}
            \end{align}
            Furthermore, if $\overline{S}$ is nonnegative and $K_T>0$, then
            \begin{align}
                c_t\leq   S_t -\overline{D}_{t,\overline{T}\given t},\label{callbound}
            \end{align}
            and if $\overline{D}_{T,\overline{T}\given T}>0$, then this is a  tighter bound.
            
            \item[(v)] Assuming $S$ is nonnegative for the upper bound, the European put satisfies
            \begin{align*}
                ( Ke^{-r(T-t)} - S_t + \overline{D}_{t,T\given t})^+\leq p_t\leq  Ke^{-r(T-t)}.
            \end{align*}
            Furthermore, if  $\overline{S}$ is nonnegative and $K_T>0$, then
            \begin{align*}
                p_t \leq K_Te^{-r(T-t)}
            \end{align*}
            and if $\overline{D}_{T,\overline{T}\given T}>0$ then this is a  tighter bound.
            
            \item[(vi)] If $r>0$, the only times at which it may be optimal to exercise an American call are just prior to  a dividend time. 
            
            \item[(vii)] If $S_{t_i} = S_{t_i-}-D_i$ and $D_i\leq K(1-e^{-r(t'_{i+1}-t'_{i})})$, where $t_i<T$, $t'_{i}=t_i\wedge T$ and $t'_{n+1}=T$, then it is not optimal to early exercise an American call just prior to the dividend time $t_i$.
        \end{enumerate}
        
    \end{proposition}
    
    \begin{proof}
        \emph{(i).}   In the escrowed dividend model with  $K_T\leq 0$, noting that the proof of Theorem \ref{mainthm} remains valid but with \eqref{line3} instead justified by the value of a forward contract, we obtain the first statement. Noting that it is optimal to exercise the call at time $T$ if and only if  $S_T-K=\overline{S}_T-K_T>0$, we obtain the second statement.

        \emph{(ii).}   By using the risk-neutral pricing formula, \eqref{pp}, and the fact that $\wt{ \overline{S}}$ is a $\wh\PP$-martingale, we have
        \begin{align}
            F(t,T)= \wh\EE[S_T\given\FFF_t] = e^{r(T-t)} \overline{S}_t +\overline D_{T,\overline{T}\given T}= (S_t-\overline{D}_{t,T\given t})e^{r(T-t)}.\label{fp2}
        \end{align}
        
        \emph{(iii).}   Applying the risk-neutral pricing formula to the identity $(S_T-K)^+-(K-S_T)^+  = (S_T-K)$ and using \eqref{fp2} gives the result.
        
        \emph{(iv)--(v).}   Applying the risk-neutral pricing formula to the inequality $(s-k)\leq (s-k)^+\leq s$ for $s=S_T$ and $k=K$ in \eqref{line1}, evaluating this using \eqref{fp2}, and noting $c_t\geq 0$ gives \eqref{usualbound}, while doing the same but for $s=\overline{S}_T\geq0$ and $k=K_T>0$ in \eqref{line2} gives \eqref{callbound}, which gives (iv). Similarly, (v)  is obtained.

        \emph{(vi).}  Let $c_t(T)$ be the price at time $t$ of a European call as a function of the maturity time $T$, and $c_t(T-)$ be its left-hand limit. Suppose that $t$ is not immediately  before  a dividend time and that there exists a next dividend time $t_i \in(t,T]$.   Then the continuation value of the American call at time $t$ is
        \begin{align*}
            C_t \geq c_t(t_i-)\geq S_t -Ke^{-r(t_i-t)}> S_t-K,
        \end{align*}
        by the lower bound in \eqref{usualbound} and  $r>0$. Thus, it is not optimal to early exercise the American call at time $t$. If there is no next dividend time such that $t_i \in(t,T]$, then the same argument applies with $t_i-$ and $t_i$ replaced by $T$.
        
        \emph{(vii)}   Now consider exercise immediately  before  dividend time $t_i$.  Suppose there exists another dividend time such that $t_{i+1} \in(t_i,T]$. Then the continuation value of the American option just prior to time  $t_i$ satisfies
        \begin{align}
            C_{t_i} \geq c_{t_i}(t_{i+1}-)\geq S_{t_i} -Ke^{-r(t_{i+1}-t_i)} = S_{t_i-}-D_i -Ke^{-r(t_{i+1}-t_i)},\label{neverexcond}
        \end{align}
        by the lower bound in  \eqref{usualbound} and $S_{t_i} = S_{t_i-}-D_i$. Now if  $D_i\leq K(1-e^{-r(t_{i+1}-t_{i})})$ holds, then algebraic manipulation of \eqref{neverexcond} yields $ C_{t_i} \geq S_{t_i-}-K$, so it is not optimal to early exercise just prior to time $t_i$.  If there is no  dividend time  such that $t_{i+1} \in(t_i,T]$, then the same argument applies with $t_{i+1}-$ and  $t_{i+1}$ replaced by $T$.
    \end{proof}

    \begin{remark}
        Proposition \ref{narel} (i) says that in the $K_T \leq 0$ case, the European call is effectively a forward. The put-call parity implies that in the Black-Scholes escrowed dividend model, the European put price at time 0 is
        \begin{align*}
            \hspace{-0.5em}  p_0 = \begin{cases}
                ( K-\overline{D}_{T}  ) e^{-rT}N(-d_2) -(S_0 -\overline{D}_{0})N(-d_1)  & \text{  \hspace{-1em}   if $K > \overline{D}_{T}$,}\\
                0& \text{  \hspace{-1em}   if $ K \leq \overline{D}_{T} $.}
            \end{cases}
        \end{align*}
        
        Surprisingly, the inclusion of positive dividends after maturity leads to tighter call and put upper bounds. The arbitrage strategy to exploit the failure of \eqref{callbound} is exactly the same as the strategy for the failure of the upper bound in \eqref{usualbound}.
        
        In  (vii), note firstly that $r\in\RR$ is permitted, secondly that $S_{t_i} = S_{t_i-}-D_i$ is an additional modelling assumption which may fail if the stock price has a random jump at time $t_i$, and thirdly that $t_i'$ is defined to account for both the cases where there is and is not another dividend time between $t_i$ and $T$.
        \halmos
    \end{remark}

    \begin{remark}\label{pcfail}
        Proposition \ref{narel} reassures us that in the escrowed dividend model, the  standard no arbitrage results for discrete dividends found in standard textbooks remain valid when there are dividends after maturity.    But these results do not hold  in this form for the piecewise GBM model, where instead they need to be expressed in terms of $\wh\EE[D_i\given\FFF_t]$, which depends on the stock price model, rather than $D_i$. For example, see \cite[Proposition 2]{HauHauLew03} for the put-call parity.  This difference, due to $D_i$ being random instead of  constant, is discussed in \cite{BurGuo12}. \halmos
    \end{remark}

    \begin{example}\label{eg1cont}
        Continuing on from Example \ref{eg1}, note that for $D_2 >  101.7512$,  the incorrect price violates the tighter call upper bound \eqref{callbound} and this bound coincides with $\sup_{\sigma\in(0,\infty)} c_0(\sigma)$, where $c_0(\sigma)$ is the price of the  European call as a function of the volatility $\sigma$, and hence it has no implied volatility  despite 
        also satisfying the usual upper bound  \eqref{usualbound}. This threshold is shown by the red vertical line in Figure \ref{fig}. \halmos
    \end{example}
    
    \begin{remark}\label{equivlem}
        In \cite{Bue10}, it is shown that any stock price paying discrete dividends must be of the form 
        \begin{align}
            S_t = (F(0,t)-\overline{D}_t)X_t + \overline{D}_t,  \quad t\in[0,\overline{T}], \label{bumodel}
        \end{align}
        for some $\wh\PP$-martingale $(X_t)_{t\in[0,\overline{T}]}$ with $\wh\EE[X_t]=1$, and this model is used in \cite{Kla15}. Since $F(0,t)-\overline{D}_t= \overline{S}_0e^{rt}$ by \eqref{fp}, and $X_t =\wt{ \overline{S}}_t/ \overline{S}_0$ by \eqref{pp}, which is without loss of generality, it follows that this is the same as the escrowed dividend model. Note that \cite{Bue10,Kla15} consider infinitely many dividends over an infinite time horizon $\overline{T}=\infty$. If the escrowed dividend model in Section \ref{sec2} were defined for $\overline{T}=\infty$, both models would still be the same. \cite{Bue10} also includes  further generalisations, such as  credit risk, which we do not address.

        In  \cite{BerMai15}, by starting with the Korn-Rogers model \cite{KR05} in which the stock price is the present value of all future expected dividends, it is shown that
        \begin{align}
            S_t = e^{-r(\overline{T}-t)}\wh\EE[S_{\overline{T}}\given \FFF_t ] +  \overline{D}_t, \quad t\in[0,\overline{T}]. \label{bmmodel}
        \end{align}
        By \eqref{fp2}, $e^{-r(\overline{T}-t)}\wh\EE[S_{\overline{T}}\given \FFF_t ]=\overline{S}_t$, so this model is also the same as the  escrowed dividend model, and \eqref{bumodel} and \eqref{bmmodel} are equal. Thus, the dynamics \eqref{bmmodel} hold without  assuming the Korn-Rogers model. Moreover, the assumption seems inconsistent with models of stocks that pay no dividend having nonzero prices.\\ \phantom{.} \halmos

    \end{remark}

    \section{Application to the Roll-Geske-Whaley formula} \label{sec3}
    
    Throughout, we assume the Black-Scholes escrowed dividend model holds. Consider an American call with strike price $K>0$ and maturity time $T\in(0,\overline{T}]$  on a stock with  a single discrete dividend before maturity, which means that $t_1<T$, and $T<t_2<\dots< t_n \leq \overline{T}$ if there are dividends after maturity.
    
    For $k_{t_1}, k_T\in\RR$, we introduce the notation
    \begin{gather*}
        d(k_{t_1}, k_T) :={}k_T-k_{t_1},\quad
        l(k_{t_1}, k_T):={}k_T(1-e^{-r(T-t_1)}), \quad
        u(k_{t_1}, k_T):={}k_T,
    \end{gather*} 
    and define the continuous function $ \lambda:= \lambda_{k_{t_1}, k_T } :(0,\infty)\to \RR$ by
    \begin{align*}
        \lambda_{k_{t_1}, k_T }  (s) =  s-k_{t_1} -c^*(s,k_T,t_1,T).
    \end{align*}
    Let $s^*:=s^*(k_{t_1}, k_T )$ denote the unique solution to $ \lambda_{k_{t_1}, k_T }(s)=0$ provided that it exists.
    
    Now for the specific values
    \begin{align*}
        K_{t_1}:=K-\overline{D}_{t_1-}, \quad K_T:=K-\overline{D}_{T}, 
    \end{align*}
    where it  should be noted that
    \begin{align}
        \overline{D}_{t_1-}= D_1+\overline{D}_Te^{-r(T-t_1)},\label{divrel}
    \end{align}
    we denote the above quantities  evaluated at $k_{t_1}=K_{t_1}$, $k_T=K_T$  as
    \begin{gather}
        D:=d(K_{t_1},K_T),\quad
        L:=l(K_{t_1},K_T),\quad
        U:=u(K_{t_1},K_T),\label{this}\\
        S^*:=s^*(K_{t_1},K_T),\quad
        \Lambda:=\lambda_{K_{t_1},K_T},\nonumber
    \end{gather}
    where $\Lambda$ is the difference between the exercise and continuation values just prior to time $t_1$. While $U=K_T$ always holds, it improves clarity to use separate variables depending on its role.

    Furthermore, assume  $r>0$. By Proposition \ref{narel} (vi), the only time at which it may be optimal to early exercise is just prior to $t_1$, so that the price of the American call at time 0 is 
    \begin{align}   
        C_0 &= \operatorname*{sup}_{\tau\in\TTT} \wh\EE[e^{-r\tau}{(S_\tau-K)^+}] \nonumber \\
        &= e^{-rt_1}  \wh\EE[({S}_{t_1-} -K)^+ \vee c_{t_1} ]\nonumber\\
        &=   e^{-rt_1}  \wh\EE[(\overline{S}_{t_1} -K_{t_1}) \vee c^*(\overline{S}_{t_1},K_T, t_1,T) ],\label{rnam}
      \end{align}
    where   $\TTT$ is the set of stopping times taking values in $[0,T]$ and $c_{t_1}$ is given by \eqref{main}.

    Accordingly, we have the trichotomy (N), (M), (A), where it is never, may be, is always optimal to early exercise just prior to time $t_1$, respectively. Any mention of early exercise will always refer to just prior to time $t_1$. The case (N) further has the dichotomy (NM), (NA), where it may be, is always optimal to exercise at time $T$, respectively. Recall that if there is indifference between early exercising and not early exercising, it is considered the latter, thereby maintaining the trichotomy. We also state these indifference cases below.

    \subsection{Early exercise lemma}

    We begin with  a generic result about the function $s\mapsto \lambda(s)$ which will be applied to characterise the early exercise policy.
    
    \begin{lemma}\label{philem} Let $t_1<T$ and  $k_{t_1},k_T\in\RR$.
        \begin{enumerate}
            \item[(i)]  If $k_T > 0$, then $\lambda$ is strictly increasing, and additionally:
            \begin{enumerate}
                \item [(a)] if $d\leq l$, then $\lambda<0$;
                \item [(b)] if $l<d<u$, then there exists a solution $s^*$ such that $\lambda(s)>0$ if and only if $s>s^*$;
                \item [(c)] if $d\geq u$, then $\lambda >0$.
            \end{enumerate}
            
            \item[(ii)]  If $k_T\leq 0$, then $\lambda= d-l$ is constant.
            
            \item[(iii)] \phantom{.}\vspace{-1em} 
            \begin{align*}
                ( s -k_{t_1} )\vee c^*(s,k_T, t_1,T) =\begin{cases}
                    s -k_{t_1} & \text{if $\lambda ( s) > 0$,}\\
                    c^*(s,k_T, t_1,T) & \text{if $\lambda ( s)\leq 0$.}
                \end{cases}
            \end{align*}
               \end{enumerate}
    \end{lemma}
    
    \begin{proof}
         \emph{(i).}    Consider the case $k_T>0$. Then $\lambda'(s)=1-N(d_1(s,k_T,t_1,T))>0$, and
        \begin{align*}
            \underline{\lambda}:=  \lim_{s\to 0^+} \lambda(s) = -k_{t_1},\quad  \overline{\lambda}:= \lim_{s\to\infty}\lambda(s) = d-l,
        \end{align*}
        where the latter is obtained by noting that
        \begin{align*}
            \lambda(s) = k_Te^{-r(T-t_1)}-k_{t_1} - (c^{\name{ BS}}(s,k_T,t_1,T)-(s - k_Te^{-r(T-t_1)}))
        \end{align*}
        and using the fact that the difference between  the standard Black-Scholes formula and the standard call lower bound vanishes as $s\to\infty$.

        Now $d\leq l$ is equivalent to $\overline{\lambda}\leq 0$, which proves (a). Next,  $d<u$ is equivalent to $-k_{t_1}<0$, which is equivalent to $\underline{\lambda}<0$, while $l<d$ is equivalent to $\overline{\lambda}>0$, which proves (b). Lastly, $d\geq u$ is equivalent to $-k_{t_1}\geq 0$, which is equivalent to $\underline{\lambda}\geq0$, which proves (c).
        
        \emph{(ii).}   Now consider the case $k_T\leq 0$.  The result  follows from the definitions of the functions. 
        
        \emph{(iii).} The result follows by noting that $\lambda ( s)<0$  is equivalent to $ s -k_{t_1} < c^*(s,k_T, t_1,T)$, and that this also holds with $<$ replaced by $=$ and $>$.
    \end{proof}

    We now apply Lemma \ref{philem} to fully characterise the optimal exercise policy of the American call in the Black-Scholes escrowed dividend model.
    
    \begin{proposition}\label{exresult}
        Suppose the Black-Scholes escrowed dividend model with $r>0$ holds and there is  a single dividend before maturity.
         
        \begin{enumerate}
                         \item[(i)] The following optimal exercise policy holds:
                \begin{enumerate}
                    \item[(N)] If $D\leq L$, then it is never optimal to early exercise. In addition,
                    \begin{enumerate}
                          \item[(NM)]   if $K_T>0$, then   exercise at maturity has probability strictly less than 1;
                        \item[(NA)] if $K_T\leq 0$, then exercise at  maturity has probability 1.
                    \end{enumerate}
                    \item[(M)] If $L <D< U$, then it is optimal to early exercise if and only if  $\overline{S}_{t_1}>S^*$, where $S^*$ exists.
                    \item[(A)] If $K_T>0$ and  $D \geq U$, or $K_T\leq 0$ and $D> L$, then it is always optimal to early exercise.
                \end{enumerate}
                
                However,  $K_T\leq 0$ and $D=L$, or $L <D< U$ and $\overline{S}_{t_1}=S^*$ if and only if there is indifference between early exercising and not early exercising.
                
                \item[(ii)] The above conditions for the cases (N), (M), (A) are mutually exclusive and exhaustive.
                
            \end{enumerate}
        \end{proposition}

        \begin{proof}

            \emph{(i).}      We apply   Lemma \ref{philem} with the  substitutions $s=  \overline{S}_{t_1}$, $k_{t_1} = K_{t_1}$, $k_{T} = K_T$. Thus, if $D\leq L$, then 
            \begin{align*}
                \Lambda(   \overline{S}_{t_1} )<0 \quad \Leftrightarrow \quad
                \overline{S}_{t_1} -K_{t_1} < c^*(\overline{S}_{t_1},K_T, t_1,T) \quad \Leftrightarrow \quad
                {S}_{t_1-} -K < c_{t_1}
            \end{align*}
            almost surely, which proves (N).  Noting that it is optimal to exercise at time $T$ if and only if  $S_T-K=\overline{S}_T-K_T>0$, we also obtain (NM) and (NA).   Similarly, if $L<D<U$, then $K_T>0$, and then
            \begin{align*}
                \overline{S}_{t_1}>S^* \quad \Leftrightarrow \quad
                \overline{S}_{t_1} -K_{t_1} > c^*(\overline{S}_{t_1},K_T, t_1,T) \quad \Leftrightarrow \quad
                {S}_{t_1-} -K> c_{t_1},
            \end{align*}
            almost surely, which proves (M). The case (A) follows similarly as the proof of (N).
            
            Lastly, there is indifference if and only if  $\Lambda(\overline {S}_{t_1})=0$. When 
            $K_T>0$, this is equivalent to  $L <D< U$ and $\overline{S}_{t_1}=S^*$, while when $K_T\leq 0$, it is equivalent to $D=L$.

            \emph{(ii).} If  $K_T>0$, we have  $0< L< U$, giving the partition   $D\leq L$, $L<D<U$, $D\geq U$, while
            if  $K_T\le0$, we have the partition $D\le L$, $D>L$. It then follows from $L<D<U$ implying $K_T> 0$ that the 3 cases are mutually exclusive and exhaustive.
        \end{proof}

        Unfortunately, Proposition \ref{exresult} is almost but not quite model-free since it relies on  properties of the Black-Scholes formula, such as in the proof of Lemma \ref{philem} (i), but it should be expected that those properties  hold for other models as well.

        By \eqref{divrel}, $D_1 \leq K(1-e^{-r(T-t_1)})$  is equivalent to  $D\leq L$, so the condition for when it is never optimal to early exercise from   Proposition  \ref{narel} (vii) and  Proposition \ref{exresult}  are the same. Note however that if $K_T\leq 0 $ and $D= L$, there is  indifference between early exercising and not. Also,    Proposition  \ref{narel} (vii) is more general as it allows for multiple dividends before maturity.

        \subsection{The always early exercise case} 
        
        In this section, we restrict to the classical case where $n=1$, so there are no dividends after maturity. Then     $\overline{S}_0 =S_0 - D_1e^{-rt_1}$ and the well-known \emph{Roll-Geske-Whaley formula} is 
        \begin{align}
            C_0 =C^{\name{RGW}}(\overline{S}_0,K_{t_1},K_T), \label{rgw}
        \end{align}
        where
        \begin{align*}
            C^{\name{RGW}}(s,k_{t_1} ,k_{T})     :={}& s N(b_1) + s  M\left(a_1,-b_1, -\sqrt{\frac{t_1}{T}}\right)  \\
            &   - k_T e^{-rT} M \left(a_2,-b_2, -\sqrt{\frac{t_1}{T}}\right)   - k_{t_1} e^{-rt_1} N(b_2),   \\
            a_i  &:= d_i(s,k_T,0,T),\\
            b_i &:= d_i(s,s^*(k_{t_1},k_T),0,t_1),\quad i=1,2,
        \end{align*}
        and $M(x,y,\rho)$ is the standard bivariate normal cdf with correlation $\rho$ at $(x,y)$, and this holds provided that $S^*$ exists \cite{Rol77,Ges79,Wha81} (see \cite[Section 9.4]{Hau07} for the  standard presentation). We choose to write the formula emphasising substitutions because this  is the canonical form that generalises to the case where there are dividends after maturity. Here, 
        \begin{align}
            K_{t_1}=K-D_1, \quad  K_T=K>0,\label{kn1}
        \end{align}
        so all the terms in the formula are well-defined and finite.
        
        It is well-known that if $D_1\leq K(1-e^{-r(T-t_1)})$, then  $S^*$ does not exist and it is never optimal to early exercise (see  \cite[pg 252]{Rol77}, \cite[Technical Note 4]{hull}, \cite[Section 9.4]{Hau07}, or  \cite[Section 5.1.5]{Kwo08}). Here, we show that there is another case, namely $D_1 \geq  K$, in which $S^*$ does not exist, and this corresponds to always early exercising, and in particular, $D_1> K(1-e^{-r(T-t_1)})$ does not guarantee the existence of $S^*$.

        \begin{proposition}\label{rgwprop}
            Suppose the Black-Scholes escrowed dividend model with $r>0$ holds, $n=1$, and there is  a single dividend before maturity. If $D_1 \geq K$, then a solution $S^*$ does not exist, it is always optimal to early exercise,  and the price of the American call at time 0 is
            \begin{align}
                C_0 =   S_0 -Ke^{-rt_1}.\label{finalcase}
            \end{align}

        \end{proposition}
        
        \begin{proof}
            Observing that \eqref{kn1} holds, $D_1\geq K$ is equivalent to $D\geq U$, and this implies by Proposition \ref{exresult} that it is always optimal to early exercise.

            Next, Lemma \ref{philem} (i)(c) and (iii) with the substitutions $s=  \overline{S}_{t_1}$, $k_{t_1} = K_{t_1}$, $k_{T} = K_T$ imply that $S^*$ does not exist and the American call price given by  \eqref{rnam} reduces to $ C_0 = e^{-rt_1}\wh\EE[\overline{S}_{t_1}-K_{t_1}]$, which gives \eqref{finalcase}.
        \end{proof}

        \begin{remark}\label{caserem}
            While \cite[pg 253--254]{Rol77}  discusses how it is always optimal to exercise in an asymptotic sense if the call is extremely deep in the money, our $D_1\geq  K$ condition is different and not asymptotic. In particular, it is  independent of $S_{t_1-}$ and depends only on the model parameters. We are not aware of it appearing in the literature.
            
            The code implementing the RGW formula given in \cite[Section 9.4]{Hau07} amazingly still produces the correct price despite neglecting this case. This is a consequence of the implementation of the numerical method for solving $S^*$, which uses a bisection method applied to the strictly negative function $-\Lambda(s)$ over the interval $s\in[0,M]$ for some large right endpoint $M>0$. As $-\Lambda(0)$ is always closer to 0 than $-\Lambda(M)$,  in this case the method always converges to $S^*=0$ (which is not a solution of $\Lambda(s)=0$ as there is no solution). When substituted into the RGW formula, this corresponds to always early exercising, hence producing the correct price. Other valid implementations may not achieve this. In contrast, the Matlab implementation given by the function \verb|optstockbyrgw| produces an error stating that $S^*$ cannot be found. \halmos
        \end{remark}
        
        Summarising what is already known in conjunction with Proposition \ref{rgwprop}, we have that if   $D_1\leq K(1-e^{-r(T-t_1)})$, then $S^*$ does not exist, it is never optimal to early exercise; if   $K(1-e^{-r(T-t_1)})< D_1< K$, then $S^*$ exists, it is optimal to early exercise if and only if $S_{t_1-} > S^*+D_1$; if   $D_1 \geq K$, then $S^*$ does not exist, it is always optimal to early exercise. In all cases, by \eqref{main}, \eqref{rgw}, \eqref{finalcase}, the price of the American call at time 0 is 
        \begin{align}
            C_0 = C^*(\overline{S}_0,K_{t_1},K_T), \label{ggbasic}
        \end{align}
        where 
        \begin{align}\label{gg}
            C^*(s,k_{t_1},k_{T}) :=   \begin{cases}
                c^*(s,k_{T},0,T) & \text{if  $d \leq l$} , \\
                C^{\name{RGW}}( s,k_{t_1}, k_{T})  & \text{if  $l<  d< u$}, \\
                s -k_{t_1}e^{-rt_1} &\text{if $k_T>0$, $d \geq u$, or   $k_T\leq 0$, $d > l$}.
            \end{cases}
        \end{align}
        Here,
        \begin{align}
            D = D_1, \quad  L = K(1-e^{-r(T-t_1)}) , \quad U = K,\label{notthis}
        \end{align}
        and note that $K_T>0$.

        \subsection{RGW formula with dividends after maturity}
        
        We drop the assumption that $n=1$, and  return to the case where there may be dividends after maturity.  We now derive an extension of the RGW formula  by using a change of variables technique similar to the proof of Theorem \ref{mainthm} in terms of the function \eqref{gg}, where it should be noted that upon substituting in $k_{t_1}=K_{t_1}$, $k_T=K_T$, we do not necessarily get \eqref{notthis}, but rather we use \eqref{this}.

        \begin{theorem}\label{theorem2}
            In the Black-Scholes escrowed dividend model with $r>0$ and a single dividend before maturity, the price at time 0 of an American call  with strike price $K>0$ and maturity time $T\in(0,\overline{T}]$ is
            \begin{align}
                C_0 = C^*(\overline{S}_0,K_{t_1},K_T ),\label{extrgw}
            \end{align}
            where $\overline{S}_0 = S_0 - \overline{D}_{0}$, $K_{t_1}=K-\overline{D}_{t_1-}$ and $K_T=K-\overline{D}_{T}$.  Furthermore, the optimal exercise policy is given by Proposition \ref{exresult}.
        \end{theorem}
        
        \begin{proof} Writing the RGW formula as a generic identity, we have for any $r>0$, $\sigma>0$, $k_{t_1}\in\RR$,  $k_T\in\RR$, $0<t_1<T$, if $S'\sim \name{GBM}(r,\sigma)$ with initial value $S'_0>0$, then
            \begin{align}
                e^{-rt_1} \EE[({S}'_{t_1}- k_{t_1})^+ \vee c^*({S}'_{t_1},k_{T},t_1,T)] = C^*(S'_0,k_{t_1},k_{T}). \label{genrgw}
            \end{align}
            To see this,  note that \eqref{rnam} and \eqref{gg} only  imply that \eqref{genrgw} holds for $k_T>0$. We separately have for $k_T\leq0$ that Lemma \ref{philem} (ii), (iii) imply
            \begin{align*}
                (s- k_{t_1})^+ \vee c^*(s,k_{T},t_1,T)  = \begin{cases}
                    s-k_{t_1}  & \text{if  $d > l$} , \\
                    s-k_Te^{-r(T-t_1)} & \text{if  $d \leq  l$} , \\
                \end{cases}
            \end{align*}
            which gives
            \begin{align*}
                e^{-rt_1} \EE[({S}'_{t_1}- k_{t_1})^+ \vee c^*({S}'_{t_1},k_{T},t_1,T)] = \begin{cases}
                    S'_0-k_{t_1}e^{-rt_1}  & \text{if  $d > l$} , \\
                    S'_0-k_Te^{-rT}    & \text{if  $d \leq  l$} , \\
                \end{cases}
            \end{align*}
            and thus  \eqref{genrgw} holds for $k_T\in\RR$.
            
            Finally, the American call price is given by  \eqref{rnam}, where the expectation is evaluated by applying  \eqref{genrgw} under $\wh\PP$  with the  substitutions $S'=\overline{S} \sim\name{GBM}(r,\sigma)$, $k_{t_1} = K_{t_1}$ and $k_{T} = K_T$, which produces \eqref{extrgw}.
        \end{proof}

        Next, we can express the optimal exercise policy primarily in terms of the signs of the modified strike prices $K_T$ and $K_{t_1}$.  
        \begin{corollary}\label{allcasescor}
            In the setting of Proposition \ref{exresult}, the following optimal exercise policy holds:
            \begin{enumerate}
                 \item[(i)]  If   $K_T>0$ and $K_{t_1}> 0$, then (NM) holds when $D\leq L$ and (M) holds when  $D> L$.
                \item[(ii)] If   $K_T> 0$ and $K_{t_1}\leq 0$, then (A) holds.
                \item[(iii)] If $K_T\leq 0$ and $K_{t_1}<0$, then  (NA) holds when $D\leq L$ and (A) holds when $D>L$.
                \item[(iv)] If $K_T\leq 0$ and $K_{t_1}\geq 0$, then (NA) holds.
                     \end{enumerate}
        \end{corollary}
        
        \begin{proof}         \emph{(i)--(ii).} First consider the case  $K_T>0$. Note the partition in the proof of Proposition \ref{exresult}, $D\leq L$, $L < D< U$, $D \geq U$, which corresponds to (NM), (M), (A), respectively.  We have that $D\leq L$ implies $K_{t_1}>0$;   $L < D < U$ implies $K_{t_1}> 0$; and $ D\geq  U$  implies $K_{t_1} \leq 0$. Summarising these results gives (i) and (ii).

            \emph{(iii)--(iv).}  Next consider the case $K_T\leq 0 $. We similarly have the partition $D\leq   L$ and $D >L$, which corresponds to (NA) and (A), respectively. We have that $D\leq L$ can imply $K_{t_1}< 0$ or $K_{t_1}\geq 0$; and  $D > L$ implies $K_{t_1}<0$. Summarising these results gives (iii) and (iv).
        \end{proof}
        
        Corollary \ref{allcasescor} also allows us to write out the extended  RGW formula in the various cases more explicitly.

        \begin{corollary}\label{cor2}
            In the setting of Theorem \ref{theorem2},   the American call price at time 0 is
                      \begin{align*}
                C_0 = \begin{cases}
                    (S_0 -\overline{D}_{0} ) N(d_1) -( K-\overline{D}_{T})e^{-rT}N(d_2) & \text{if $K_T >0, D\leq L$,} \\
                    C^{\name{RGW}}( \overline{S}_0,K_{t_1}, K_{T}) & \text{if $K_T >0, L<D<U$,} \\
                    S_0  -\overline{D}_{0,{T}\given 0} -(K-D_1)e^{-rt_1}&  \begin{array}{@{}l@{}}
                        \text{if $K_T >0, D\geq U$, }\\
                        \text{or $K_T \leq 0, D>L$,}
                    \end{array} \\
                    S_0 -\overline{D}_{0,T\given 0} -Ke^{-rT} & \text{if  $K_T \leq 0, D\leq L$,}
                \end{cases}
            \end{align*}
            where $d_1$, $d_2$ are given in \eqref{d1d2},
            and if $n=1$, then the price reduces to \eqref{ggbasic}.
        \end{corollary}

        It is clear that $C_0$ is well-defined and finite. For instance,  $C^{\name{RGW}}(\overline{S}_0,K_{t_1},K_{T})$ is only defined for $K_{T},K_{t_{1}}>0$, but this is ensured by the condition $L< D<U$.     Also, despite assuming $K>0$ throughout, clearly the pricing formulas \eqref{main} and \eqref{extrgw} and the optimal exercise policy in Proposition \ref{exresult} are still valid for $K\leq 0$,  though the upper bounds in Proposition \ref{narel} (iv)--(v) are not.

    \begin{example}
        Consider the situation in Example \ref{eg1}  but with $D_1,D_2$  given in Table \ref{table} and for American calls  with  $K=80$, $T=1$ instead. 
        Observing that Corollary \ref{allcasescor} breaks down into 6 cases when counting the 2 cases in (i) and (iii),  Table \ref{table} shows that all cases are nonempty and gives the corresponding optimal exercise policy and American call price by  \eqref{extrgw}.  In the (M) case, $S^*=98.4489$.  \\ \phantom{.}\halmos

        \begin{table}[htbp]
            \centering
            \centerline{
                \begin{tabular}{
                        c
                        c
                        S[table-format=-1.2]
                        S[table-format=-1.2]
                        S[table-format=-1.2]
                        S[table-format=-1.2]
                        S[table-format=-1.2]    
                        c
                        S[table-format=2.5]
                    }
                    \hline
                    $D_1$ & $D_2$ & {$K_{t_1}$} & {$K_T$} & {$L$} & {$D$} & {$U$} & {Exercise Policy} & {$C_0$} \\
                    \hline
                    2 & 5  & 73.29  & 75.15  & 2.22   & 1.86   & 75.15  & (NM) & 25.0294 \\
                    3 & 5  & 72.29  & 75.15  & 2.22   & 2.86   & 75.15  & (M)  & 24.3629 \\
                    3 & 82 & -0.22  & 0.42   & 0.01   & 0.65   & 0.42   & (A)  & 22.3644 \\
                    3 & 83 & -1.17  & -0.55  & -0.02  & 0.62   & -0.55  & (A)  & 22.3644 \\
                    2 & 83 & -0.17  & -0.55  & -0.02  & -0.38  & -0.55  & (NA) & 22.7179 \\
                    1 & 83 & 0.83   & -0.55  & -0.02  & -1.38  & -0.55  & (NA) & 23.6884 \\
                    \hline
            \end{tabular}}
            \caption{For different dividend amounts $D_1$, $D_2$ corresponding to the cases in Corollary \ref{allcasescor}, the optimal exercise policy  and the American call price $C_0$ are given.}
            \label{table}
        \end{table}
        
    \end{example}

    \begin{example}\label{exrgw1}
              
        Consider the situation in Example \ref{eg1}  but for an American call with $K=80$, $T=1$ instead. 
        Then this is the case (M) if $D_2<e^{r(t_2-t_1)}(K-D_1)=81.7614$ and (A) if $D_2\geq 81.7614$. In Figure \ref{fig2},  this threshold is shown by the black vertical line, and these cases correspond to $K_{t_1}>0$ and $K_{t_1}\leq 0$, respectively.  Figure \ref{fig2} shows the price of the American call $C_0$  computed  by the correct formula \eqref{extrgw} and by incorrectly  using  \eqref{rgw}, where only the dividends before maturity are included. The corresponding  implied volatility is also shown, computed using the correct formula.

        \begin{figure}[H]
            \begin{center}
                \includegraphics{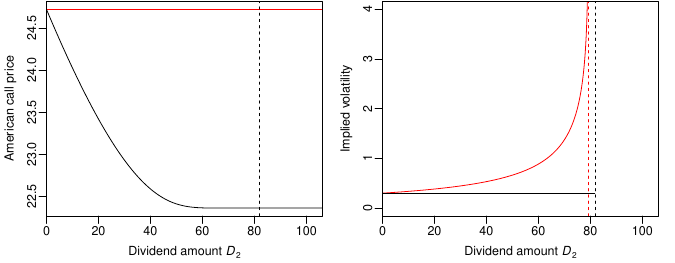}
                \caption{The American call price $C_0$ (left) and implied volatility (right) as a function of the dividend amount $D_2$ computed by the correct formula \eqref{extrgw} (black) and   incorrect formula \eqref{rgw} (red).}\label{fig2}
                
            \end{center}
        \end{figure}
        
        Note that for $D_2\geq 79.1840$, no implied volatility exists for the incorrect price, and this threshold is shown by the red vertical line.   For $D_2\geq  81.7614$, the option price is independent of $\sigma$, so the implied volatility does not uniquely exist even for the correct price. However, for $D_2$ less than but close to $81.7614$, the implied volatility for the correct price is of course $\sigma$, but  cannot be computed reliably due to numerical instability. \halmos

           \end{example}

    \begin{example}
          Consider the situation in Example \ref{eg1}  but with $D_2=5$  and for American calls  with  $K=80$, $T=1$ instead. 
        We construct an $N$-step approximating binomial tree for the stock price $S=\overline{S}+\overline{D}$ up to time $T$ using the method of \cite[Section 23.1]{hull}, in which the usual binomial  tree is constructed  for $\overline{S}$ and then $\overline{D}$ is added to each node by \eqref{pp}.  Figure \ref{fig3} shows the call prices computed using the $N$-step tree converging to the extended RGW formula \eqref{extrgw} as $N\to\infty$. The latter is the same as the (M) case in Table \ref{table}. This numerically demonstrates that the pricing formula with dividends after maturity is  correct in this example. Other American and exotic options can be priced consistently on this tree or its extension up to time $\overline{T}$.  \halmos

        \begin{figure}[htpb]
            \begin{center}
                \includegraphics{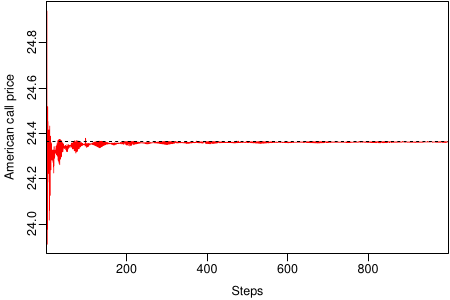}
                \caption{The  call price $C_0$ computed using an $N$-step binomial tree as a function of $N=1,2,\dots, 1000$ (red)  and by the extended RGW formula  \eqref{extrgw} (black).}\label{fig3}
            \end{center}
        \end{figure}
    \end{example}

    \subsection{Claims of arbitrage and related inconsistencies} \label{secarb}
    In \cite[Section 9.4]{Hau07}, it is claimed that the RGW formula \eqref{rgw} has ``considerable flaws that lead to significant arbitrage opportunities, and thus renders it more or less useless for all practical purposes." The details in \cite{Hau07,HauHauLew03} are as follows.
    
    \begin{example}\label{abrex}
        Consider the situation in Example \ref{eg1}  but with $D_1=7$, $t_1=0.9999$, $D_2=0$ instead. All strike prices are $K=130$.
        An American call with maturity  $T=1$ just after the dividend time has price 4.3007 using the RGW formula \eqref{rgw}, while the European call with maturity  $T = 0.9998$ just before the dividend time has price  4.9183 using \eqref{bs}. This is claimed to be an arbitrage, as the American call, which has more exercise rights than the European call, should not be worth less than it. However, while there are no dividends after maturity for the American call,  there is one for the European call, so the  European call price is actually 4.3001 by \eqref{main}, and so this is not an arbitrage. \\ \phantom{.} \halmos
    \end{example}

    A similar arbitrage claim about the RGW formula is made in \cite[pg 116]{Fri02} with a different numerical example where the  3 call prices above are 88,173.32, 96,348.77 and  88,173.32, respectively. In pricing these, we have interpreted just prior and just after the dividend time as the limit, and the second price slightly differs from the printed value 96,197 as \cite{Fri02} used a numerical method which, neglecting discretisation error, should be  the same as our use of \eqref{bs}. This is not an arbitrage for the same reason. Similar incorrect arbitrage claims are made in \cite[Sections 1--2]{BenVor02} and \cite[pg 109--110]{BosGaiShe03}. Overall, such arguments about inconsistencies in the escrowed dividend model and disfavouring its use arise from applying pricing formulas and methods that are from inconsistent models. For example, \eqref{bs} comes from a  Black-Scholes escrowed dividend model that holds up to time $T$, leading to different models over different maturities as discussed in Remark \ref{rem1}.  But by first fixing the Black-Scholes escrowed dividend model for a   time horizon $\overline{T}$ and then pricing options for any maturity $T\in(0,\overline{T}]$ within this model, such as by using \eqref{main}, \eqref{extrgw}, or Monte Carlo methods for exotic options, the issue is resolved.  Moreover, the model necessarily has no arbitrage due to the existence of a risk-neutral measure.

    \begin{remark}\label{ogrem2} Both Example \ref{abrex} and the above point essentially appear in \cite[Example 2]{BerMai15} and the preceding text. In their case, all 3 calls are American and they use a tree-based method for pricing, whereas we are able to use pricing formulas instead, despite them being American.  Similar points about the above resolution are also made in \cite[pg 206--207]{BenVor02}, \cite[Section 4.2.5]{Kla15}, \cite[pg 266]{VelNi06}.  \halmos
    \end{remark}

    The issue of inconsistent models  similarly arises with Black's approximation. For simplicity, assume there is a single dividend before maturity and $n=1$, then Black's approximation is $C_0 \approx \max(c_0(t_1-),c_0(T))$ (this notation is defined in the proof of Proposition \ref{narel} (vi)), where the method describes that the 2  European call prices are computed using \eqref{bs} \cite{black75}. It has been noted that Black's approximation cannot overestimate $C_0$ since the American call can in actuality be exercised at more than just these 2 times (see \cite[Equation (5.1.31)]{Kwo08}, \cite[pg 211]{Rol77}). But it has also been noted that Black's approximation can overestimate $C_0$. We observe  here that this is only possible because the incorrect pricing formula \eqref{bs} is used, but not when the correct formula \eqref{main} is used, as expected. In \cite[Exercise 15.28]{hull} and its answer, the overestimation is mentioned, but not this explanation for it.

    \section{Conclusion} \label{sec4}

    Within  the escrowed dividend model, we have derived standard no arbitrage results and extensions of the Black-Scholes and RGW formulas when there are dividends after maturity, including  fully characterising the optimal exercise policy for the latter. This package of results follows from a single coherent fixed model that does not admit arbitrage, contrary to claims otherwise.
    
    We hope this improves both the practice and teaching of option pricing. For practice, we have derived analytical pricing formulas in a consistent way across different maturities, and demonstrated that  failing to account for dividends after maturity can lead to reasonably large errors in some cases, but small errors in many cases. For teaching, which was our original motivation, we have provided a  framework that is arguably the most natural and simple extension of the Black-Scholes model to the discrete dividend setting. Within this model, results such as put-call parity and the Black-Scholes formula for discrete dividends, in the form they are presented in  standard textbooks, are provable and consistent, with extensions to dividends after maturity, thereby improving pedagogy. This contrasts with the patchwork of methods and models in the literature, including the piecewise GBM model, in which standard no arbitrage results take a different form and there are no known analytical pricing formulas. For these reasons, we advocate the use of the escrowed dividend model.
    
    Lastly, the change of variable technique here should be applicable to extend pricing formulas in non-Black-Scholes   escrowed dividend models to the case where there are dividends after maturity.

    \bibliographystyle{alpha}  
    \bibliography{bibliography}  
    
\end{document}